\documentclass[%
 reprint,
 amsmath,
 amssymb,
 aps,
 unsortedaddress
]{revtex4-2}
\usepackage{tikz}
\usetikzlibrary{shapes,arrows,positioning,fit,backgrounds}
\usepackage{quantikz}
\usepackage{subcaption}
\usepackage{graphicx}
\usepackage{dcolumn}
\usepackage{bm}
\usepackage{makecell}
\usepackage{multirow}
\usepackage{booktabs}
\usepackage{amsthm,easybmat,verbatim}
\usepackage{xcolor}
\definecolor{lightblue}{rgb}{0.68, 0.85, 0.9}
\usepackage{listings}
\usepackage{hyperref}
\usepackage{float}
\hypersetup{
    colorlinks,
    linkcolor=blue,
    citecolor=blue,
    urlcolor=blue
}

 \newtheorem{corollary}{Corollary}
 
\newtheorem*{prop*}{Proposition*}

\usepackage{orcidlink}

\usepackage[ruled, linesnumbered]{algorithm2e}

\usepackage{physics}

\usepackage[normalem]{ulem}

\usepackage{colortbl}

\newtheorem{definition}{Definition}

\newtheorem{lemma}{Lemma}

\newtheorem{remark}{Remark}

\begin{document}

\preprint{APS/123-QED}

\title{A matching decomposition algorithm for simulating quantum walk Hamiltonians}

\author{Mostafa Atallah$^{1,3}$\orcidlink{0009-0004-8187-6932}}
\author{Alvin Gonzales$^2$\orcidlink{0000-0003-1635-106X}}
\author{Daniel Dilley$^2$\orcidlink{0000-0002-8821-4059}}
\author{Igor Gaidai$^{1,4}$\orcidlink{0000-0002-3950-3356}}
\author{Zain H. Saleem$^2$\orcidlink{0000-0002-8182-2764}}
\author{Rebekah Herrman$^1$\orcidlink{0000-0001-6944-4206}}\thanks{corresponding author}
\email{rherrma2@utk.edu}
\affiliation{$^1$Department of Industrial and Systems Engineering, University of Tennessee Knoxville, USA}
\affiliation{$^2$Mathematics and Computer Science Division, Argonne National Laboratory, Lemont, IL, USA}
\affiliation{$^3$Department of Physics, Faculty of Science, Cairo University, Giza 12613, Egypt}
\affiliation{$^4$Department of Physics and Astronomy, University of Tennessee Chattanooga, USA}

\date{\today}

\begin{abstract}
In this work, we present a new algorithm for generating quantum circuits that efficiently implement continuous time quantum walks on arbitrary simple sparse graphs.  
The algorithm, called matching decomposition, works by decomposing a continuous-time quantum walk Hamiltonian into a collection of exactly implementable Hamiltonians corresponding to matchings in the underlying graph followed by a novel graph compression algorithm that merges edges in the graph. We develop a greedy matching heuristic and a compression-aware matching heuristic, both of which can be used in the quantum circuit algorithm. Lastly, we convert the walks to a circuit and Trotterize over these components.
The dynamics of the walker on each edge in the matching can be implemented in the circuit model as sequences of CX and CRx gates. We do not use Pauli decomposition when implementing walks along each matching. 
Furthermore, we compare greedy (compression-aware) matching decomposition to a standard Pauli-based simulation pipeline and find that greedy (compression-aware) matching decomposition consistently yields substantial resource reductions, requiring up to 43$\%$ (70\%) fewer controlled gates and up to 54$\%$ (75\%) shallower circuits than Pauli decomposition across multiple graph families. 
Finally, we also present examples and theoretical results for when matching decomposition can exactly simulate a continuous-time quantum walk on a graph. 
\end{abstract}

\maketitle

\section{Introduction}\label{sec:intro}
Hamiltonian simulation is a central primitive in quantum computing, and a wide variety of techniques have been developed for it. 
Some Hamiltonian simulation approaches rely on structured access to the Hamiltonian, for example through oracles \cite{low2018hamiltonian, low2017optimal, berry2009black}, compressed quantum states \cite{somma2025shadow}, or block encodings \cite{chakraborty2018power}. 
Other techniques excel on particular kinds of Hamiltonians such as sparse Hamiltonians \cite{berry2015hamiltonian} or Hamiltonians where the underlying problem has some symmetry \cite{qiang2016efficient}. 

A Continuous-Time Quantum Walk (CTQW) is a quantum process whose dynamics can be described by $e^{-i A t}$ where a Hamiltonian $A$ is an adjacency matrix of some unweighted graph and $t$ is the propagation time. Thus, we can equivalently say that CTQW is a process driven by a given graph.

CTQWs are known to form a universal model of quantum computation \cite{childs2009universal, childs2004spatial, herrman2019continuous} and have a wide range of applications, including spatial search \cite{chakraborty2020finding, osada2020continuous, apers2022quadratic, tanaka2022spatial}, link prediction \cite{goldsmith2023link, moutinho2023quantum}, quantum simulation \cite{schreiber20122d}, optimization \cite{marsh2020combinatorial, slate2021quantum}, variational quantum algorithms \cite{matwiejew2024quantum}, quantum state preparation \cite{gonzales2025efficient}, biological processes \cite{d2021protein, santiago2020quantum}, and more \cite{wang2020experimental, wang2022continuous}. 
Although CTQWs can be implemented using waveguides \cite{rai2008transport, perets2008realization} or photonic chips \cite{chapman2016experimental, qiang2016efficient, tang2018experimental}, only CTQWs on specific classes of graphs, such as circulant graphs, stars, and graph products, can be exactly implemented in the circuit model \cite{qiang2016efficient, portugal2022implementation, adhikari10circuit, loke2017efficient, qu2022deterministic, douglas2009efficient}. CTQWs on other classes of graphs can be approximated by decomposing the adjacency matrix (or Laplacian) of the underlying graph into smaller matrices and Trotterizing over the smaller matrices \cite{farhi1998quantum, chakraborty2025continuous}.

CTQWs on dynamic graphs provide a more direct bridge to the standard circuit model \cite{herrman2019continuous}. 
In this framework, computation is performed by evolving under a time-ordered sequence of adjacency matrices, and the resulting dynamics constitute a universal model of quantum computation \cite{herrman2019continuous}.
Subsequent work has extended this framework in several directions, including the addition of isolated vertices \cite{wong2019isolated}, simplification techniques \cite{herrman2022simplifying}, methods to write an arbitrary gate in terms of at most three graphs \cite{adisa2021implementing}, and a detailed characterization of the relationship between CTQWs on dynamic graphs and the multi-angle quantum approximate optimization algorithm (MA-QAOA) \cite{herrman2022relating}, which connects CTQWs on dynamic graphs to the broader context of QAOA-like variational algorithms \cite{farhi2014quantum, kazi2025analyzing, wilkie2024quantum, zhao2025symmetry}.

More recently, CTQWs on dynamic graphs were used to implement a sparse state preparation circuit using $O(nm)$ CX gates, where $n$ is the number of qubits and $m$ is the number of nonzero amplitudes in the target state \cite{gonzales2025efficient}. 
The work relies on converting single edges and self-loops in the dynamic graph framework to sequences of CRx, CP, and CX gates, thus motivating the idea that arbitrary CTQWs can be simulated in the circuit model using similar decomposition approaches. 
The authors show that these sequences can generate any unitary \cite{gonzales2025efficient}. 
Related quantum-walk–based state preparation schemes have also been proposed in other settings, for example, to prepare angular momentum eigenstates \cite{shi2024preparing}. 
Other decomposition-based approaches include decomposing graphs into star subgraphs and performing Trotterization over the stars \cite{Chen_2024} and decomposing the Laplacian of a graph into smaller pieces to use for Trotterization \cite{chakraborty2025continuous}.

Another variation of CTQW is a CTQW on \textit{state space}, where the graph corresponding to the walk is defined directly on a state space. The key difference is that in this formulation we are required to respect a particular embedding of the graph into the state space when simulating the walk, which is important when CTQW is used as an intermediate building block of a quantum circuit. Note that, given an unlabeled graph, one can label the vertices of the graph in order to perform a CTQW on a state space.

In general, implementing an arbitrary CTQW requires exponentially many quantum gates and classical computational resources in the number of qubits. However, a sparse CTQW, that is a CTQW where the number of edges grows polynomially in the number of qubits, can be implemented using fewer resources. 

In this work, we develop an algorithm that generates a quantum circuit for a given sparse CTQW on state space.
Our algorithm is based on decomposition of CTQW edges into a collection of matchings. 
We develop two matching heuristics that can be used in the quantum circuit algorithm: one is a greedy heuristic, and one is a compression-aware heuristic.
A novel graph compression algorithm is used to combine edges, which reduces CX gate overhead in the output circuit. The overall evolution can then be approximated via Trotterization over the set of matchings \cite{Suzuki1976_GenTrottFormAndSystApprox}.  
As with any Trotterization methods, the approximation error can be arbitrarily reduced by increasing the number of Trotter steps.

In our numerical benchmarks we compare against a standard Pauli decomposition Trotterization workflow with identical transpilation settings. We emphasize that Pauli-based Hamiltonian simulation is a mature direction, and a number of specialized Pauli evolution compilation techniques (e.g., cancellation-aware synthesis and dedicated Pauli compilers) can further reduce two-qubit gate counts \cite{liu2025quclear}. Our goal here is to evaluate whether matching decomposition provides a new circuit construction primitive that is competitive against a widely used baseline, and to quantify the practical gains that can be obtained immediately in a standard toolchain.

In Sec.~\ref{sec:background}, we present a formal definition for a CTQW on a dynamic graph sequence and define the basic graph theory terminology used in this work. 
In Sec.~\ref{sec:algorithm}, we describe the greedy matching and compression-aware matching decomposition algorithms. 
In Sec.~\ref{sec:results}, we compare the 2-norm of the operator difference between the exact CTQW time evolution operator and both of the matching decompositions, as well as Pauli decomposition \cite{Pauli} and find that all decompositions have comparable accuracy. We find that the greedy matching decomposition approach requires up to 43\% fewer CX gates and up to 54\% shallower circuits than Pauli decomposition and that the compression-aware matching requires up to 70\% fewer CX gates and up to 75\% shallower circuits than Pauli decomposition when used to approximate CTQWs on select graphs. We also analyze when a given graph has a matching decomposition that pairwise commute.
In Sec.~\ref{sec:discussion}, we discuss the implications of our findings and potential future directions.

\section{Background}\label{sec:background}
In this section, we give a brief introduction to continuous-time quantum walks on dynamic graphs, graph theory terminology, Trotterization, and Pauli decomposition. 

\subsection{Continuous-time quantum walks on dynamic graphs}\label{sec:dynamicctqw}
A \textit{continuous-time quantum walk on a dynamic graph}, $\textsc{G} = \{(G_i, t_i)\}_{i \in [m]}$ is defined by an ordered sequence of graphs $G_i$ and associated propagation times $t_i$. 
A walker traverses each graph for the associated time according to 
\begin{equation*}
    \ket{\psi_f} = e^{-i A_m t_m} \ldots e^{-i A_1 t_1} \ket{\psi_0}
\end{equation*}
where $\ket{\psi_0}$ is the initial state of the walker and $A_i$ is the adjacency matrix of graph $G_i$. 

The authors of \cite{gonzales2025efficient} showed how to map CTQWs on single edge walks and self-loop graphs into the circuit model using CRx, CRz, and CP gates. In this approach, each vertex of the graph is labeled by a bitstring of length $n$. We make extensive use of the single edge conversion. The single-edge conversion will rely on the \textit{Hamming distance} between two labels. The Hamming distance of two bitstrings $u$ and $v$ is the number of positions in which the two bitstrings differ and is denoted $d_H(u,v)$. Given an edge between two standard basis states $z_1$ and $z_2$ that are Hamming distance one apart, the circuit is a multicontrolled Rx gate,
\begin{align}
    \text{Rx}(2t)=\begin{pmatrix}
        \cos(t) &-i\sin(t)\\
        -i\sin(t) &\cos(t)
    \end{pmatrix},
\end{align}
with a target on the differing qubit and controls matching the other qubits. For $z_1$ and $z_2$ greater than Hamming distance one apart, a similar approach can be applied. First, a basis transformation is performed. CX gates are applied such that the control is a fixed qubit where $z_1$ and $z_2$ differ and the targets are the remaining differing qubits. This brings $z_1$ and $z_2$ to Hamming distance one apart. Second, the multicontrolled Rx gate is applied as before.  Finally, the CX gates are applied again to restore the basis.

These circuits were then used to make an efficient sparse deterministic state preparation algorithm that also utilized ideas from classical combinatorics to minimize the controlled gate count. 
The Hamiltonian simulation algorithm in this paper is inspired by a multi-edge version of the CTQW mapping from \cite{gonzales2025efficient}. 

\subsection{Graph theory terminology}\label{sec:graphterminology}
Throughout this work, we will focus on \textit{simple} graphs $G = (V,E)$ that have no self-loops or multi-edges.
A \textit{matching} of $G$ is defined to be a set of edges $M \subseteq E$ such that no two edges in $M$ share a common vertex. 
A graph matching is said to be \textit{maximal} if it is not a subset of any other matching of $G$. 
A matching $M$ of $G$ is said to be \textit{perfect} if $|M| = |V|/2$. 
Graph matchings have applications in fields such as image analysis \cite{bunke1995efficient}, chemical compound analysis \cite{smalter2008gpm}, and computer science \cite{karp1990optimal, ren2021matching}. 
Furthermore, there exist several graph matching algorithms and heuristics that can be implemented in polynomial time \cite{edmonds1965paths, gabow1976implementation, hopcroft1973n, micali1980v, blum1990new, gabow1991faster, gabow1991faster, micali1980v}.
Note that a graph can have several different matchings, so the matching algorithms used in this paper are not deterministic in general. 
Thus, there may be some slight differences in circuits when using the Trotterization methods in this work.

A graph is said to be \textit{bipartite} if its vertices can be partitioned into two disjoint sets such that every edge has an endpoint in both sets, or equivalently, the graph contains no cycles of odd length. Two graphs commonly discussed in later sections are complete graphs on 2 vertices (equivalently single edges), denoted $K_2$, and cycles with four vertices, denoted $C_4$, both of which are examples of bipartite graphs. A \textit{cycle} on $n$ vertices is a connected graph with $n$ vertices where each vertex is incident to exactly two edges. Note that $C_n$ is bipartite if and only if $n$ is even. $K_{n,n}$ is another such bipartite graph mentioned below and is defined as a graph with $2n$ vertices that are partitioned into two sets of size $n$ such that all vertices in one set are connected to all vertices in the other set by edges and no vertices are connected to any vertex in the same set.

\subsection{Trotter Decomposition}
An arbitrary Hamiltonian $A_G$ can be decomposed into a set of matrices $\{A_{i}\}$. Then, by the Trotter-Lie formula~\cite{lapidus1981generalization},
\begin{equation}
    \lim_{N \rightarrow \infty} \left(e^{-i \frac{t}{N}A_{1}} e^{-i \frac{t}{N}A_{2}} \cdots e^{-i \frac{t}{N}A_{m}}\right)^N = e^{-i t A_G}
\end{equation}
with error $O(t^2/N)$ for finite $N$. This approach is commonly used to simulate Hamiltonians on gate-based quantum computers.

\section{Matching decomposition algorithm}
\label{sec:algorithm}

In this section, we describe our approach for implementing continuous-time quantum walks (CTQWs) on sparse graphs embedded into a state space in the circuit model. The method consists of two main steps: (1) decomposing a graph into matchings, and (2) applying iterative graph compression on the matchings to minimize the number of control qubits required for circuit implementation.

\subsection{Greedy matching decomposition}

Since edges in a matching are independent (i.e., no two edges share a common vertex), time evolution operators under each edge in the matching commute, and the total matching evolution $e^{-iA_M t}$ can be decomposed into independent single-edge operations. Since the embedding of the graph is fixed, each vertex is labeled with an $n$-bit string corresponding to a basis state $|z\rangle$. 

Our algorithm decomposes the graph edges into matchings by grouping edges according to their \textit{bit-flip structure}. 
The bit-flip structure of an edge $(u, v)$ is determined by which bit positions differ between the endpoint labels $u$ and $v$. 
We first identify edges with Hamming distance $d_H(u,v) = 1$. 
For such edges, we define the \textit{bit-flip position} as the unique bit index $i$ where $u$ and $v$ differ. 
For example, the edge $(010, 110)$ has bit-flip position 2, since only the third bit (index 2) differs between the endpoints.
Here and further in the text, the bit indices are assigned from right to left and are indexed starting with ``0".

The key insight is that walks along edges sharing the same bit-flip position $i$ all correspond to controlled Rx rotations with target qubit $i$ and different sets of control bits. 
This makes them natural candidates for grouping into the same matching, as the resulting circuit can potentially be compressed using the graph compression technique described in Section~\ref{sec:space_reduction}.

The greedy matching algorithm (Algorithm~\ref{alg:greedy-matching}) partitions edges into matchings as follows: edges with Hamming distance 1 are grouped by their bit-flip position. Crucially, edges with the same bit-flip position cannot share a vertex -- if two edges $(u_1, v_1)$ and $(u_2, v_2)$ both flip bit $i$, then sharing a vertex would imply $u_1 = u_2$ (or some permutation), which forces $v_1 = v_2$ since both flip only bit $i$. Thus, each bit-flip group is already a valid matching. 
 Edges with Hamming distance greater than one are greedily added to existing matchings if they share no vertices with any edge already in an existing matching in order to reduce the total number of matchings, but they benefit from graph compression in a different way.  See Sec.~\ref{sec:example} for an example of the matching algorithm. 

\begin{algorithm}
\small
\caption{Greedy Matching Decomposition}
\label{alg:greedy-matching}
\SetAlgoLined
\KwIn{Edge set $E$ with $n$-bit vertex labels}
\KwOut{List of matchings $\mathcal{M} = \{M_1, M_2, \ldots\}$}

\tcp{Group edges by bit-flip position}
\For{each edge $(u,v) \in E$}{
    \eIf{$d_H(u,v) = 1$}{
        $i \leftarrow$ bit position where $u$ and $v$ differ\;
        Add $(u,v)$ to group $G_i$\;
    }{
        Add $(u,v)$ to $E_{\text{multi}}$ \tcp*{Hamming dist $> 1$}
    }
}
$\mathcal{M} \leftarrow \emptyset$\;

\tcp{Create matchings from each bit-flip group}
\For{each bit position $i$ with $G_i \neq \emptyset$}{
    $\mathcal{M} \leftarrow \mathcal{M} \cup \{G_i\}$\;
}

\For{each edge $e \in E_{\text{multi}}$}{
    Place $e$ in first $M \in \mathcal{M}$ with no vertex conflict, or create new $M$\;
}
\Return{$\mathcal{M}$}\;
\end{algorithm}

\subsection{Iterative graph compression}
\label{sec:space_reduction}

After decomposing the graph into matchings, we apply an iterative graph compression
technique to each matching in the decomposition. The graph compression combines edges of the graph, which consequently reduces the number of edges and both the number of multicontrolled Rx gates and controls in the multicontrolled Rx gates of the output circuit. Recall
from Section~\ref{sec:dynamicctqw} that each edge $(z_1, z_2)$ in a graph is
implemented as a multicontrolled Rx gate: the target qubit is one of the bit positions
where $z_1$ and $z_2$ differ, and the control qubits are determined by the common bits of $z_1$ and $z_2$. When multiple bit positions differ (i.e., Hamming distance greater than one), we select any one of the differing positions as the target qubit and apply controlled-NOT gates to the other differing positions, as described in Section~\ref{sec:dynamicctqw}. The goal of graph compression is to merge multiple edges into a single edge in a reduced qubit space, thereby reducing the number of control qubits needed.

Let $M = \{e_1, \ldots, e_m\}$ be a matching where each edge $e_i = (u_i, v_i)$ connects vertices labeled by $n$-bit strings $u_i, v_i \in \{0,1\}^n$. We refer to the graph at iteration zero as $G_0$.

To describe the compression algorithm precisely, we introduce several definitions. First, we need a way to identify which bit positions differ between the endpoints of an edge, as this determines the target qubit for the corresponding Rx gate.

\begin{definition}[XOR Mask]
    For an edge $(u, v)$ with $n$-bit vertex labels, the \textit{XOR mask} is $\mu = u \oplus v \in \{0,1\}^n$, where bit $k$ of $\mu$ is 1 if and only if $u$ and $v$ differ at position $k$. We distinguish between the \textit{original XOR mask} (computed from the original $n$-bit labels, which never changes during compression) and the \textit{current XOR mask} (computed from the current compressed labels).
\end{definition}

As edges are merged, the bit-string labels become shorter. We need to track which original qubit indices remain in the compressed representation.

\begin{definition}[Active Qubits]
    The \textit{active qubits} $\mathcal{A} = (a_0, a_1, \ldots, a_{m-1})$ is an ordered list of original qubit indices that remain after edge merging. Initially, $\mathcal{A} = (0, 1, \ldots, n-1)$. When two edges are merged by removing the bit at position $p$ in the current compressed labels, the corresponding original qubit index $a_p$ is removed from $\mathcal{A}$. The compressed edge labels are $m$-bit strings, where $m = |\mathcal{A}|$, and bit position $p$ in the compressed label corresponds to original qubit $a_p$.
\end{definition}

When a qubit is removed during compression and that qubit corresponded to a differing bit in the original edge (i.e., the original XOR mask has a 1 at that position), we must track this for circuit construction. These qubits require additional controlled-NOT gates to implement the basis transformation described in Section~\ref{sec:dynamicctqw}.

\begin{definition}[Weight-Reducing Qubits]
    The \textit{weight-reducing qubits} $\mathcal{W}_{e}$ is an ordered list of original qubit indices that were removed during compression and correspond to differing bits in the original edge. Initially, $\mathcal{W}_e$ is empty. When two edges are merged by removing the bit at position $p$, if the original qubit index $a_p$ corresponds to a 1 in the original XOR mask $\mu_e$ (i.e., bit $a_p$ is set in $\mu_e$), then $a_p$ is appended to $\mathcal{W}_{e'}$ of the merged edge. These qubits are called ``weight-reducing'' because removing them reduces the Hamming weight of the current XOR mask relative to the original.
\end{definition}

With these concepts in place, we can now define when two edges can be merged into a single edge in a reduced qubit space.

\begin{definition}[Mergeable edges]
Two edges $e_1 = (u_1, v_1)$ and $e_2 = (u_2, v_2)$ are \textit{mergeable at position $p$} (where $p$ is a position in the current compressed labels) if:
\begin{enumerate}
    \item The edges have equal original XOR masks ($\mu_{e_1} = \mu_{e_2}$),
    \item The active qubit lists are equal ($\mathcal{A}_{e_1} = \mathcal{A}_{e_2}$),
    \item The weight-reducing qubit lists are equal ($\mathcal{W}_{e_1} = \mathcal{W}_{e_2}$), and
    \item The endpoints differ only at position $p$: either $(u_1 \oplus u_2 = v_1 \oplus v_2 = 2^p)$ or $(u_1 \oplus v_2 = v_1 \oplus u_2 = 2^p)$.
\end{enumerate}
Note that conditions 1--3 together imply the current XOR masks are also equal. The last condition ensures that after deleting bit $p$, both edges collapse to the same edge in the reduced space.
\end{definition}

The result of merging two edges is a compressed edge with fewer bits in its vertex labels.

\begin{definition}[Compressed Edge]
A \textit{compressed edge} $e' = (u', v') \in \{0,1\}^m \times \{0,1\}^m$ is an edge with $m$-bit vertex labels, where $m \leq n$. Each compressed edge has associated metadata:
\begin{itemize}
    \item $\mathcal{A}_{e'} = (a_0, \ldots, a_{m-1})$: the ordered list of active original qubit indices, where bit position $p$ in $e'$ corresponds to original qubit $a_p$,
    \item $\mathcal{W}_{e'}$: the list of weight-reducing original qubit indices (a subset of the deleted qubits), and
    \item $\mu_{e'} \in \{0,1\}^n$: the original XOR mask.
\end{itemize}
A compressed edge represents $2^{n-m}$ original edges: expanding $e'$ by inserting bits at the deleted qubit positions with all $2^{n-m}$ combinations yields the original edge set.
\end{definition}



Finally, we define the operation that removes a bit from vertex labels when merging two edges.

\begin{definition}[Bit Deletion Operator]
The \textit{bit deletion operator} $\pi_k: \{0,1\}^n \to \{0,1\}^{n-1}$ removes bit $k$ from an $n$-bit string:
\begin{equation}
\pi_k(b_{n-1} \cdots b_k \cdots b_0) = b_{n-1} \cdots b_{k+1} b_{k-1} \cdots b_0.
\end{equation}
We extend $\pi_k$ to edges by $\pi_k((u, v)) = (\pi_k(u), \pi_k(v))$.
\end{definition}

The merge operation replaces $e_1, e_2$ with a single edge $e' = (\pi_p(u_1), \pi_p(v_1))$ in the reduced space, where $p$ is the merge position in the current compressed labels.

The graph compression algorithm (Algorithm~\ref{alg:space-reduction}) works as follows. Starting with the original matching $M$ where all edges have $n$-bit labels, we initialize $M' \leftarrow M$. Each edge $e \in M'$ maintains its own active qubit list $\mathcal{A}_e$ (initially $(0, 1, \ldots, n-1)$), weight-reducing list $\mathcal{W}_e$ (initially empty), and original XOR mask $\mu_e$. The algorithm repeatedly searches for pairs of edges that can be merged according to Definition 4. When edges $e_1$ and $e_2$ are merged at position $p$, they are replaced by a single edge $e' = (\pi_p(u_1), \pi_p(v_1))$. The new active qubit list $\mathcal{A}_{e'}$ is obtained by removing the $p$-th element from $\mathcal{A}_{e_1}$. If the removed original qubit index corresponds to a bit set in $\mu_{e_1}$ (i.e., it was a differing bit in the original edge), then that index is appended to $\mathcal{W}_{e'}$.

The algorithm terminates when no more mergeable pairs exist. The output is a set of compressed edges, where each compressed edge $e'$ has its own $\mathcal{A}_{e'}$ and $\mathcal{W}_{e'}$. Edges that were never merged retain their original $n$-bit labels. See Sec.~\ref{sec:example} for an example of the graph compression algorithm. 

\begin{algorithm}
\small
\caption{Iterative Graph Compression}
\label{alg:space-reduction}
\SetAlgoLined
\KwIn{Matching $M \subseteq \{0,1\}^n \times \{0,1\}^n$}
\KwOut{Set of tuples $(e', \mathcal{A}_{e'}, \mathcal{W}_{e'}, \mu_{e'})$ for each edge}

\For{each edge $e = (u, v) \in M$}{
    $\mathcal{A}_e \leftarrow (0, 1, \ldots, n-1)$\tcp*{ordered list}
    $\mathcal{W}_e \leftarrow ()$\tcp*{empty list}
    $\mu_e \leftarrow u \oplus v$\tcp*{original XOR mask}
}
$M' \leftarrow M$\;

\Repeat{no merge performed}{
    \For{each pair $e_1, e_2 \in M'$ with $\mu_{e_1} = \mu_{e_2}$, $\mathcal{A}_{e_1} = \mathcal{A}_{e_2}$, $\mathcal{W}_{e_1} = \mathcal{W}_{e_2}$}{
        \For{$p = 0$ \KwTo $|\mathcal{A}_{e_1}| - 1$}{
            \If{$(u_1 \oplus u_2) = (v_1 \oplus v_2) = 2^p$ \textbf{or} $(u_1 \oplus v_2) = (v_1 \oplus u_2) = 2^p$}{
                $e' \leftarrow (\pi_p(u_1), \pi_p(v_1))$\;
                $\mathcal{A}_{e'} \leftarrow \mathcal{A}_{e_1}$ with element at position $p$ removed\;
                $\mathcal{W}_{e'} \leftarrow \mathcal{W}_{e_1}$\;
                $k \leftarrow \mathcal{A}_{e_1}[p]$\tcp*{original qubit index}
                \If{bit $k$ is set in $\mu_{e_1}$}{append $k$ to $\mathcal{W}_{e'}$}
                $\mu_{e'} \leftarrow \mu_{e_1}$\;
                $M' \leftarrow (M' \setminus \{e_1, e_2\}) \cup \{e'\}$\;
                \textbf{break}\;
            }
        }
    }
}
\Return{$\{(e', \mathcal{A}_{e'}, \mathcal{W}_{e'}, \mu_{e'}) : e' \in M'\}$}\;
\end{algorithm}

\subsection{Compression-aware matching decomposition}
Instead of taking a na\"ive greedy matching approach, note that one can develop a matching heuristic based on the XOR masks that were introduced in the previous subsection. The compression-aware heuristic treats all edges uniformly by grouping edges according to their full XOR mask $\mu = u \oplus v$, where $u$ and $v$ are the $n$-bit vertex labels of an edge, regardless of Hamming distance. 
Edges that share the same XOR mask differ in exactly the same bit positions, making them natural candidates for merging during the iterative graph compression step. The algorithm places each edge into a matching using a three-tier priority: (1)~prefer a matching that already contains an edge with the same XOR mask and has no vertex conflict, maximizing compression opportunities; (2)~fall back to any conflict-free matching; (3)~create a new matching as a last resort. To avoid sensitivity to a single edge ordering, the algorithm runs $T$ independent trials with different group processing
orders-largest-first, smallest-first, and random shuffles-each seeded from a provided seed list. After all trials, the matching that results in the quantum circuit with the lowest estimated CX cost (computed via a fast compression-based estimator) is returned. The pseudocode for the compression aware heuristic can be found in Alg.~\ref{alg:compression-aware}. Figure~\ref{fig:matching_comparison} illustrates how the greedy and compression-aware heuristics differ.
\begin{algorithm}
\small
\caption{Compression-Aware Matching Decomposition}
\label{alg:compression-aware}
\SetAlgoLined
\KwIn{Edge set $E$ with $n$-bit vertex labels, list of seeds $S = [s_1, s_2, \ldots, s_T]$}
\KwOut{List of matchings $\mathcal{M}^*$ minimizing estimated CX cost}

$\mathcal{M}^* \leftarrow \texttt{null}$; \quad $C^* \leftarrow \infty$\;

\For{$t \leftarrow 1$ \KwTo $T$}{
    Initialize random generator with seed $s_t$\;

    \tcp{Group \emph{all} edges by XOR mask}
    \For{each edge $(u,v) \in E$}{
        $\mu \leftarrow u \oplus v$ \tcp*{XOR mask (any Hamming dist)}
        Add $(u,v)$ to group $X_\mu$\;
    }

    \tcp{Vary group processing order across trials}
    \lIf{$t = 0$}{sort groups largest-first}
    \lElseIf{$t = 1$}{sort groups smallest-first}
    \lElse{randomly shuffle groups}

    $\mathcal{M} \leftarrow \emptyset$\;
    \For{each group $({\mu}, X_{\mu})$ in processing order}{
        Randomly shuffle edges within $X_{\mu}$\;
        \For{each edge $e \in X_{\mu}$}{
            \tcp{Priority 1: same XOR mask, no conflict}
            \lIf{$\exists\, M_j \in \mathcal{M}$: mask $\mu \in M_j$ \textbf{and} no vertex conflict}{add $e$ to $M_j$}
            \tcp{Priority 2: any matching, no conflict}
            \lElseIf{$\exists\, M_j \in \mathcal{M}$: no vertex conflict}{add $e$ to $M_j$}
            \tcp{Priority 3: create new matching}
            \lElse{$\mathcal{M} \leftarrow \mathcal{M} \cup \{\{e\}\}$}
        }
    }

    $C \leftarrow \textsc{EstimateCX}(\mathcal{M})$\;
    \lIf{$C < C^*$}{$\mathcal{M}^* \leftarrow \mathcal{M}$; \quad $C^* \leftarrow C$}
}
\Return{$\mathcal{M}^*$}\;
\end{algorithm}

\begin{figure*}
\centering
\includegraphics[width=\textwidth]
{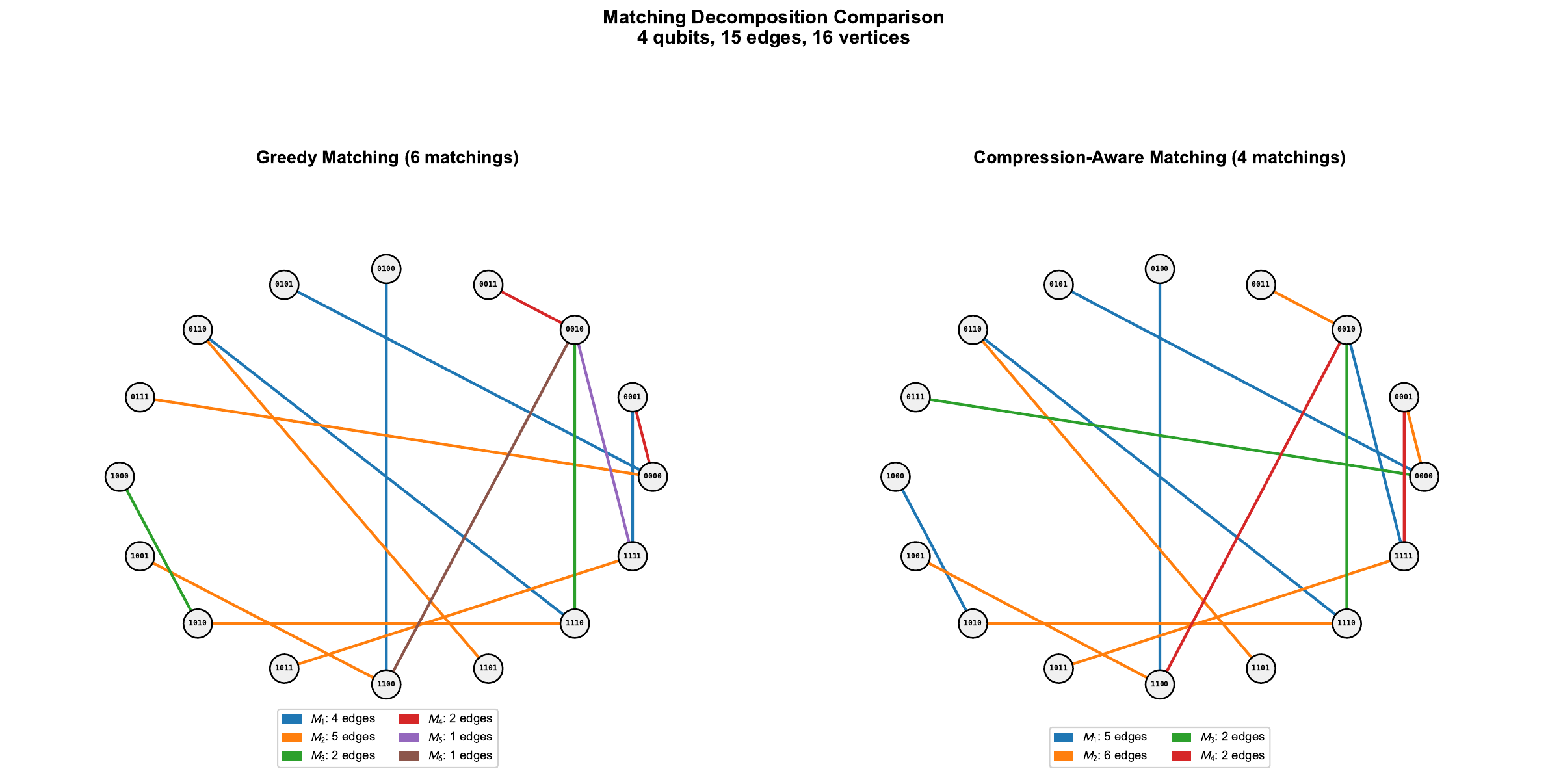}
\caption{Matching decomposition of a connected 16-vertex graph (4~qubits,
  15~edges). \textbf{Left:} Greedy heuristic (6~matchings). \textbf{Right:}
  Compression-aware heuristic (4~matchings). Edge colors indicate matching
  assignment. The compression-aware algorithm consolidates edges more
  efficiently by using XOR-mask-based grouping with a multi-trial search.}
\label{fig:matching_comparison}
\end{figure*}

\subsection{Circuit construction and Trotterization}

For circuit construction, each compressed edge $e' = (u', v')$ with metadata $(\mathcal{A}_{e'}, \mathcal{W}_{e'})$ is implemented as a three-stage circuit. Let $q_{\text{target}}$ denote the target qubit of the multicontrolled Rx gate, which is one of the bit positions where $u'$ and $v'$ differ in the compressed representation; the corresponding original qubit index is $\mathcal{A}_{e'}[q_{\text{target}}]$.

The first stage applies CX gates to transform the computational basis so that the original edge (which may have Hamming distance greater than one) becomes Hamming distance one in the transformed basis. Specifically, for each original qubit index $w \in \mathcal{W}_{e'}$, a CX gate is applied with control on original qubit $\mathcal{A}_{e'}[q_{\text{target}}]$ and target on original qubit $w$.

The second stage applies the multicontrolled $\text{Rx}(2t)$ gate that implements the edge evolution. The target of this gate is original qubit $\mathcal{A}_{e'}[q_{\text{target}}]$, corresponding to a differing bit in the compressed edge. The control qubits are all other qubits in $\mathcal{A}_{e'}$, namely $\{\mathcal{A}_{e'}[p] : p \neq q_{\text{target}}\}$, with control values determined by the common bits of $u'$ and $v'$. For fully compressed edges where the compressed labels are single-qubit strings, this reduces to a simple $\text{Rx}(2t)$ gate with no controls.

The third stage reverses the basis transformation by applying the same CX gates as in the first stage. Edges that are not compressed are implemented directly using the single-edge circuit construction from \cite{gonzales2025efficient}, which follows the same pattern: X gates to transform to the all-ones basis, CX gates for Hamming distance reduction, a multicontrolled Rx gate, and the inverse transformations.

After constructing circuits for all matchings (that may or may not contain compressed edges), we approximate the full CTQW evolution via Trotterization. For matchings $\{M_j\}_{j=1}^m$ decomposing graph $G$, we have $A_G = \sum_{j=1}^m A_{M_j}$. 

\subsection{Example}
\label{sec:example}

Let us implement the greedy matching decomposition algorithm on the graph in Fig.~\ref{fig:ctqw-reduction}. The matching algorithm first identifies that edges $00$-$01$ and $10$-$11$ have bit-flip position $0$ and no edges have bit-flip position $1$. Thus, $00$-$01$ and $10$-$11$ are added to matching $M_0$. The remaining two edges $00$-$11$ and $01$-$10$ cannot be added to $M_0$ since they share vertices with edges already contained in $M_0$. Thus we place them into a second matching called $M_1$.

After assigning each edge to a matching, we apply the iterative graph compression techniques to each matching. The edges of $M_0$ differ only in qubit~0, compressing to a single edge $(0,1)$ with $\mathcal{A} = \{0\}$ and $\mathcal{W} = \{\}$. For $M_1$, both edges $(00,11)$ and $(01,10)$ have Hamming distance 2, differing in both qubits. Iterative compression eliminates qubit~0 (which varies across the matching), yielding a compressed edge $(0,1)$ with $\mathcal{A} = \{1\}$ and $\mathcal{W} = \{0\}$. The non-empty $\mathcal{W}$ reflects that qubit~0 contributed to the Hamming distance within each original edge.

Finally, we construct the circuits for each compressed matching. The compressed edge for $M_0$ requires only a single $\text{Rx}(2t)$ gate on qubit~0. For $M_1$, the non-empty weight-reducing set $\mathcal{W} = \{0\}$ requires additional CX gates: we first apply a CX gate with control on qubit~1 (the target of the $\text{Rx}$ gate) and target on qubit~0, then apply $\text{Rx}(2t)$ on qubit~1, and finally reverse the CX gate. The complete circuit for one Trotter step is shown in Fig.~\ref{fig:example-circuit}; multiple Trotter steps are implemented by repeating this circuit sequence with appropriately scaled rotation angles.

\begin{figure*}
\centering
\begin{tikzpicture}[scale=1.0]
    \begin{scope}[shift={(0,-2.0)}]
        \node[font=\bfseries] at (1.1,3.2) {Original Graph $G$};
        \node[circle, draw, fill=blue!25, minimum size=0.7cm, font=\small] (00) at (0,1.8) {$00$};
        \node[circle, draw, fill=blue!25, minimum size=0.7cm, font=\small] (01) at (2.2,1.8) {$01$};
        \node[circle, draw, fill=blue!25, minimum size=0.7cm, font=\small] (10) at (0,0) {$10$};
        \node[circle, draw, fill=blue!25, minimum size=0.7cm, font=\small] (11) at (2.2,0) {$11$};
        \draw[thick, blue] (00) -- (11);   
        \draw[thick, blue] (01) -- (10);   
        \draw[thick, red] (00) -- (01);    
        \draw[thick, red] (10) -- (11);    
    \end{scope}

    \draw[->, thick, line width=1.5pt] (3.0,-1.1) -- (4.0,-1.1);
    \node[font=\small\bfseries] at (3.5,-0.6) {Decompose};

    \begin{scope}[shift={(5.0,0)}]
        \node[font=\bfseries, blue] at (1.1,3.2) {$M_1$};
        \node[circle, draw, fill=blue!25, minimum size=0.7cm, font=\small] (00) at (0,1.8) {$00$};
        \node[circle, draw, fill=blue!25, minimum size=0.7cm, font=\small] (01) at (2.2,1.8) {$01$};
        \node[circle, draw, fill=blue!25, minimum size=0.7cm, font=\small] (10) at (0,0) {$10$};
        \node[circle, draw, fill=blue!25, minimum size=0.7cm, font=\small] (11) at (2.2,0) {$11$};
        \draw[thick, blue] (00) -- (11);
        \draw[thick, blue] (01) -- (10);
    \end{scope}

    \draw[->, thick, line width=1.5pt] (8.0,0.9) -- (9.0,0.9);
    \node[font=\small\bfseries] at (8.5,1.4) {Compress};

    \begin{scope}[shift={(9.8,0)}]
        \node[font=\bfseries, blue] at (1.1,1.8) {Compressed $M_1$};
        \node[font=\small] at (1.1,1.3) {$\mathcal{A} = \{1\}$, $\mathcal{W} = \{0\}$};
        \node[circle, draw, fill=blue!25, minimum size=0.7cm, font=\small] (0) at (0.2,0.5) {$0$};
        \node[circle, draw, fill=blue!25, minimum size=0.7cm, font=\small] (1) at (2.0,0.5) {$1$};
        \draw[thick, blue] (0) -- (1);
    \end{scope}

    \begin{scope}[shift={(5.0,-4.2)}]
        \node[font=\bfseries, red] at (1.1,3.2) {$M_0$};
        \node[circle, draw, fill=red!25, minimum size=0.7cm, font=\small] (00) at (0,1.8) {$00$};
        \node[circle, draw, fill=red!25, minimum size=0.7cm, font=\small] (01) at (2.2,1.8) {$01$};
        \node[circle, draw, fill=red!25, minimum size=0.7cm, font=\small] (10) at (0,0) {$10$};
        \node[circle, draw, fill=red!25, minimum size=0.7cm, font=\small] (11) at (2.2,0) {$11$};
        \draw[thick, red] (00) -- (01);
        \draw[thick, red] (10) -- (11);
    \end{scope}

    \draw[->, thick, line width=1.5pt] (8.0,-3.3) -- (9.0,-3.3);
    \node[font=\small\bfseries] at (8.5,-2.8) {Compress};

    \begin{scope}[shift={(9.8,-4.2)}]
        \node[font=\bfseries, red] at (1.1,1.8) {Compressed $M_0$};
        \node[font=\small] at (1.1,1.3) {$\mathcal{A} = \{0\}$, $\mathcal{W} = \{\}$};
        \node[circle, draw, fill=red!25, minimum size=0.7cm, font=\small] (0) at (0.2,0.5) {$0$};
        \node[circle, draw, fill=red!25, minimum size=0.7cm, font=\small] (1) at (2.0,0.5) {$1$};
        \draw[thick, red] (0) -- (1);
    \end{scope}
\end{tikzpicture}
\caption{Matching decomposition and space reduction example. The original 4-vertex graph $G$ with edges $\{(00, 01), (10, 11), (00, 11), (01, 10)\}$ decomposes into two matchings. $M_0$ (red) contains Hamming distance 1 edges $\{(00, 01), (10, 11)\}$ which differ only in qubit~0, compressing to a single edge $(0,1)$ with active qubits $\mathcal{A} = \{0\}$ and weight-reducing qubits $\mathcal{W} = \{\}$. $M_1$ (blue) contains Hamming distance 2 edges $\{(00, 11), (01, 10)\}$, compressing to edge $(0,1)$ with $\mathcal{A} = \{1\}$ and $\mathcal{W} = \{0\}$. The non-empty $\mathcal{W}$ for $M_1$ indicates that CX gates are required for the basis transformation.}
\label{fig:ctqw-reduction}
\end{figure*}

\begin{figure}
\centering
$\left(\rule{0pt}{3em}\right.$%
\begin{quantikz}[row sep=0.4cm, column sep=0.4cm]
 \lstick{$q_0$} & \gate{R_x\!\left(\frac{2t}{N}\right)} & \slice{$M_0$} & \targ{} & \qw & \targ{} & \slice{$M_1$} & \qw \\
 \lstick{$q_1$} & \qw & & \ctrl{-1} & \gate{R_x\!\left(\frac{2t}{N}\right)} & \ctrl{-1} & & \qw
\end{quantikz}
$\left.\rule{0pt}{3em}\right)^{\!N}$
\caption{Trotterized circuit for the example graph with $N$ Trotter steps. The $R_x(2t/N)$ gate on $q_0$ implements $M_0$. The CX-$R_x(2t/N)$-CX sequence implements $M_1$, where the CX gates perform the basis transformation required by $\mathcal{W} = \{0\}$. The entire sequence is repeated $N$ times.}
\label{fig:example-circuit}
\end{figure}

\section{Results}
\label{sec:results} 
In this section, we first calculate the operator norm difference between CTQW unitaries and the unitary resulting from Trotterization using our matching decomposition algorithm on a test set of graphs and compare it to the operator norm difference between the CTQW unitary and Trotterization using Pauli decomposition to show that both methods converge at similar rates. 
We perform Pauli decomposition by iterating through all $4^n$ Pauli strings and calculating the coefficient $c_P = \text{Tr}(P^\dagger A) / 2^n$ for each Pauli operator $P$ and the graph adjacency matrix $A$. Note that this is a simple reference implementation and the benchmarks are the properties of the compiled circuit and not coefficient finding. Non-zero terms are collected into Qiskit's \texttt{SparsePauliOp} \cite{Pauli}, and the circuit is constructed using Qiskit's \texttt{PauliEvolutionGate} \cite{PauliEvolutionGate} for each Trotter step. Code for our matching algorithm and Pauli decomposition can be found at the linked repository below. 
We then calculate the CX gate count and circuit depth that matching decomposition requires to simulate CTQWs on test sets of graphs with between 8 and 128 vertices and compare it to the CX gate count and circuit depth that Pauli decomposition requires to simulate the same graphs. 
Finally, we analytically analyze graphs that have a commuting matching decomposition. Note that this does not imply that our matching decomposition algorithm will find this matching.

We use first-order Trotterization throughout this work for simplicity.
All circuits are constructed using Qiskit 1.2.2, full connectivity, basis of CX and $\text{U}_3$ (an arbitrary single qubit unitary), and optimized at level 3.

\subsection{Operator difference 2-norm}\label{sec:operatornorm}
We study the error induced by using both greedy and compression-aware matching decomposition as a basis for Trotterization, $\|e^{-iAt} - U_{\text{matching}}\|_2$, and compare it to the error induced by using Pauli decomposition in Trotterization by computing $\|e^{-iAt} - U_{\text{Pauli}}\|_2$, where $A$ is the graph adjacency matrix and $U_{\text{matching}}$ ($U_{\text{Pauli}}$) is the operator resulting from using Trotterization with matching (Pauli) decomposition. Figure~\ref{fig:norm_convergence_8v} shows the average 2-norm operator difference for a set of 80 randomly generated sparse connected 8-vertex graphs with mean values over all graphs and shaded $\pm 1$ standard deviation regions. Convergence plots for 200 randomly generated connected graphs on both 16- and 32-vertices, 74 randomly generated disconnected 8-vertex graphs, and 200 randomly generated disconnected graphs on both 16- and 32-vertices can be found in Appendix~\ref{sec:normplots}.

At $t=1.0$ and $N=100$, connected graphs achieve mean errors of $9.1$--$9.5 \times 10^{-3}$ with $\sim$15\% relative standard deviation across graph sizes. Disconnected graphs show lower mean errors ($6.1$--$6.7 \times 10^{-3}$) due to fewer edges and thus fewer non-commuting terms, but higher relative variance ($\sim$30\%). Some disconnected graphs achieve machine-precision errors when all matchings commute. The matching and Pauli decompositions achieve nearly identical accuracy across all cases, confirming that matching decomposition provides a valid basis for Trotterization.

\begin{figure}
\centering
\includegraphics[width=.48\textwidth]{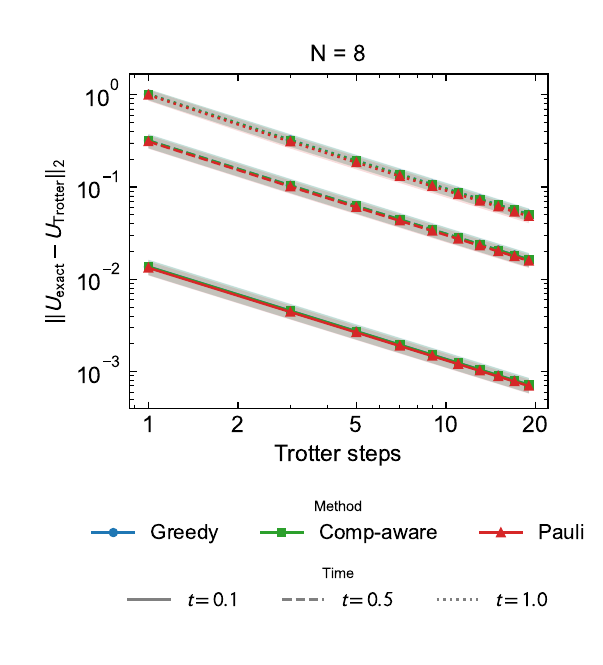}
\caption{Operator 2-norm difference vs.\ number of Trotter steps for
  connected graphs on 8~vertices (3 qubits). Three methods (greedy,
  compression-aware, Pauli) with $\pm 1\sigma$ bands across 5 runs of
  80~graphs each, at three evolution times $t = 0.1, 0.5, 1.0$
  (distinguished by line style).}
\label{fig:norm_convergence_8v}
\end{figure}

\subsection{CX gate count and circuit depth}\label{sec:gatecount}

Next, we compare the gate counts for one iteration of Trotterization using the
matching decomposition and Pauli decomposition on two datasets, implemented using Qiskit's
transpiler~\cite{qiskit_transpiler} at optimization level 3 to optimize the resulting circuits. Circuit depth is the number of layers of gates executed in parallel. The first dataset consists of connected graphs that are generated as Hamiltonian paths on $N=2^n$ vertices for $n \in \{3,4,5,6,7\}$ that visit vertices labeled in numerical order $0\ldots0 \to 0\ldots01 \to 0\ldots 010 \to \cdots \to 1\ldots1$ or slight modifications of Hamiltonian paths. See the repository linked at the end of the paper for code that generates these graphs and details on the modifications to the Hamiltonian paths. 
These graphs have $O(N)$ edges and have a characteristic Hamming weight distribution: for all $i \leq n$, approximately $\frac{N}{2^i} = \frac{2^n}{2^i}$ edges lie in $H_i$ where $H_k = \{uv \in E : d_H(u,v) = k\}$. The graphs have large diameter, low maximum degree, low edge density, and are frequently bipartite (19--60\% of graphs, with smaller graphs more frequently bipartite). Properties of graphs in this dataset are shown in Table~\ref{tab:graph_properties}.

\begin{table*}[htbp]
\centering
\caption{Properties of the connected graph datasets (mean $\pm$ std, averaged over all graphs in each dataset). Bipartite \% shows the fraction of graphs that are bipartite. Connected graphs are single-component counting paths with $\sim N$ edges.}
\label{tab:graph_properties}
\begin{tabular}{lccccc}
\toprule
\textbf{Vertices} & \textbf{Edges} & \textbf{Avg Degree} & \textbf{Max Degree} & \textbf{Density} & \textbf{Bipartite \%} \\
\midrule
\multicolumn{6}{c}{\textit{Connected datasets}} \\
\midrule
8 & $8.8 \pm 0.4$  & $2.2 \pm 0.1$ & $3.3 \pm 0.5$ & $0.32 \pm 0.01$ & 22\% \\
16 & $16.7 \pm 0.5$  & $2.1 \pm 0.1$ & $3.1 \pm 0.4$ & $0.14$ & 32\% \\
32 & $32.8 \pm 0.4$  & $2.1 \pm 0.03$ & $3.1 \pm 0.3$ & $0.07$ & 26\% \\
64 & $65.2 \pm 0.6$  & $2.0 \pm 0.02$ & $3.1 \pm 0.3$ & $0.03$ & 24\% \\
128 & $129.6 \pm 0.7$  & $2.0 \pm 0.01$ & $3.1 \pm 0.3$ & $0.02$ & 18\% \\
\bottomrule
\end{tabular}
\end{table*}

\begin{figure}[htbp]
\centering
\begin{subfigure}[t]{0.48\textwidth}
  \includegraphics[width=\textwidth]{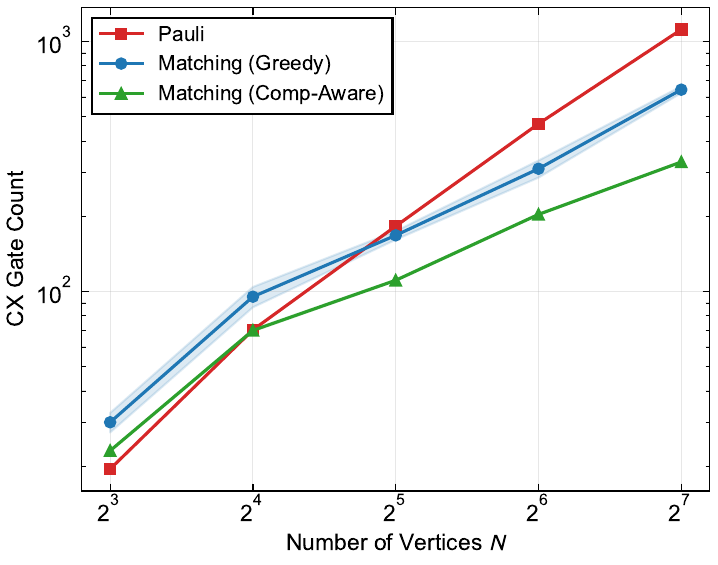}
  \caption{CX count vs.\ vertices.}
\end{subfigure}\hfill
\begin{subfigure}[t]{0.48\textwidth}
  \includegraphics[width=\textwidth]{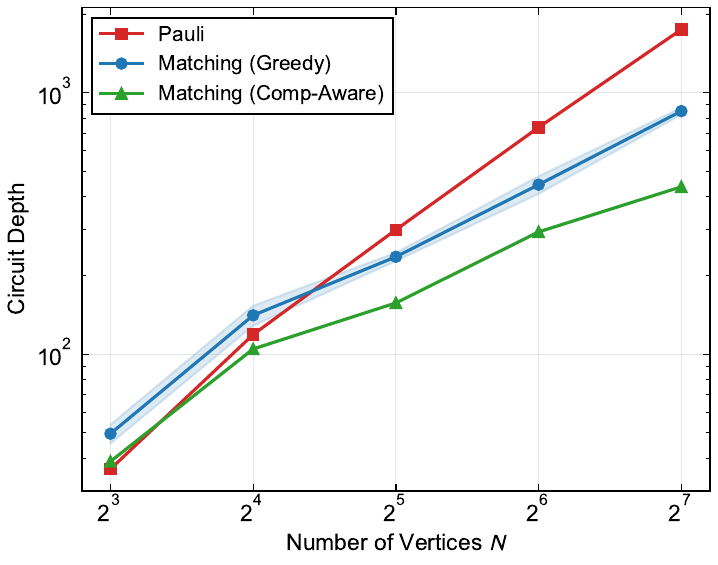}
  \caption{Circuit depth vs.\ vertices.}
\end{subfigure}
\caption{Three-way scaling comparison for connected graphs. Blue: Greedy matching,
  Green: Compression-aware matching, Orange: Pauli. Shaded regions show $\pm 1\sigma$
  across 5 runs.}
\label{fig:connected_scaling}
\end{figure}

The average CX gate counts and circuit depths are shown in Figure~\ref{fig:connected_scaling}. 
For each value of $N$, we generate 200 graphs that have the above properties and run the matching decomposition and Pauli decomposition algorithms 5 times as both the matching heuristic and Qiskit optimizer are nondeterministic. On smaller datasets (8 vertices) both greedy and compression-aware matching decomposition on average require more CX gates than Pauli decomposition.
However, the crossover occurs at 16 vertices for compression-aware matching and 32 vertices for greedy matching. Compression-aware matching requires  matching decomposition requires .5\%, 39\%, 56\%, and 70\% fewer CX gates on average, for 16-, 32-, 64-, and 128-vertex graphs.  Greedy matching requires  matching decomposition requires 8\%, 34\%, and 43\% fewer CX gates on average, for 32-, 64-, and 128-vertex graphs. Circuit depth follows a similar pattern: compression-aware matching requires circuits that are 12\%, 47\%, 60\%, and 75\% shallower than Pauli decomposition on 16-, 32-, 64-, and 128-vertex graphs while greedy matching decomposition produces 21\% shallower circuits at 32 vertices, 40\% shallower circuits at 64 vertices, and 54\% shallower circuits at 128 vertices for connected graphs.

The shaded regions in Figure~\ref{fig:connected_scaling} show the standard deviation across graphs, expressed as a coefficient of variation (CV = std/mean). For connected graphs, matching decomposition exhibits higher relative variance (CV = 35--46\%) compared to Pauli decomposition (CV = 15--22\%), reflecting matching decomposition's sensitivity to graph structure or choice of matching. Interestingly, matchings that consist purely of edges in $H_1$ (e.g., Gray code paths on hypercubes) consistently require \textit{more} CX gates to implement than Pauli decomposition does--the mixed Hamming structure is essential for reduced CX gate count. 
 Based on these observations, we believe the CX gate count when simulating CTQWs on graphs using Trotterization with matching decomposition will be lower than the CX gate count when using Trotterization with Pauli decomposition on sparse graphs where there are edges that connect vertices with large Hamming distance, but there is no regular structure to the graph. We believe this because Pauli decomposition of edges that connect vertices with large Hamming distance should require nonlocal terms, and the irregularity ensures that there are few terms in the Pauli decomposition that can be simplified. 

The second dataset consists of Erd\H{o}s-R\'enyi random graphs generated with the NetworkX
\texttt{gnp\_random\_graph}($N$, $p$) function~\cite{networkx_gnp_random_graph} for $N \in
\{8, 16, 32, 64, 128 \}$ nodes and $p =.01$ probability of an edge existing. For each $N$, we generate 100 random graphs.

Figure~\ref{fig:er_cx_scaling} shows the CX gate count and circuit depth for \texttt{gnp\_random\_graph}($N$, $.01$). The CX gate count (Fig.~\ref{fig:er_cx_scaling}a) for greedy matching decomposition is less than Pauli decomposition when $N \geq 32$, achieving reductions of 25\% at $N=32$, 33\% at $N=64$, and 31\% at $N=128$. The circuit depth (Fig.~\ref{fig:er_cx_scaling}b) shows similar trends, with greedy matching decomposition producing 37\%, 41\%, and 49\% shallower circuits at $N=32$, 64, and 128 respectively. Interestingly, the compression-aware matching decomposition results are almost identical to the greedy matching decomposition results for this dataset. The shaded regions indicate standard deviation across graphs. For this dataset, the CV for the matching decomposition CX counts decreases from 29\% at 32 vertices to 10\% at 128 vertices, indicating increasingly predictable performance for larger $N$.

\begin{figure}[ht]
\centering
\begin{subfigure}[t]{0.48\textwidth}
  \includegraphics[width=\textwidth]{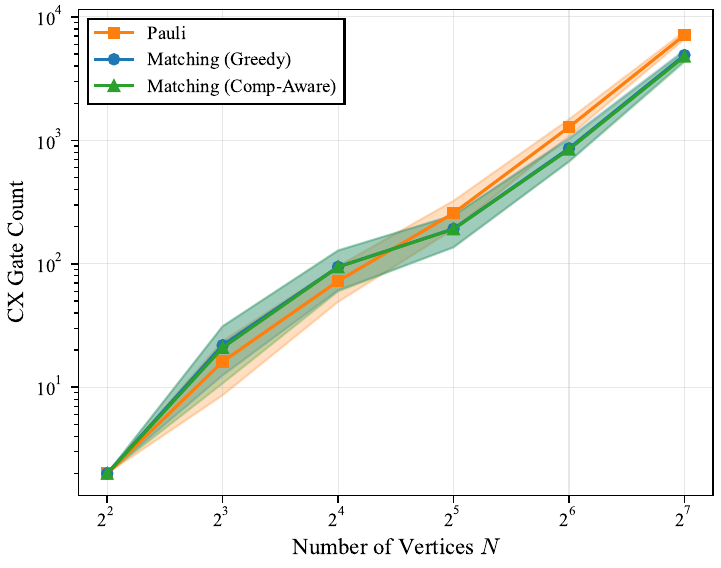}
  \caption{$p = 0.01$}
\end{subfigure}\hfill
\begin{subfigure}[t]{0.48\textwidth}
  \includegraphics[width=\textwidth]{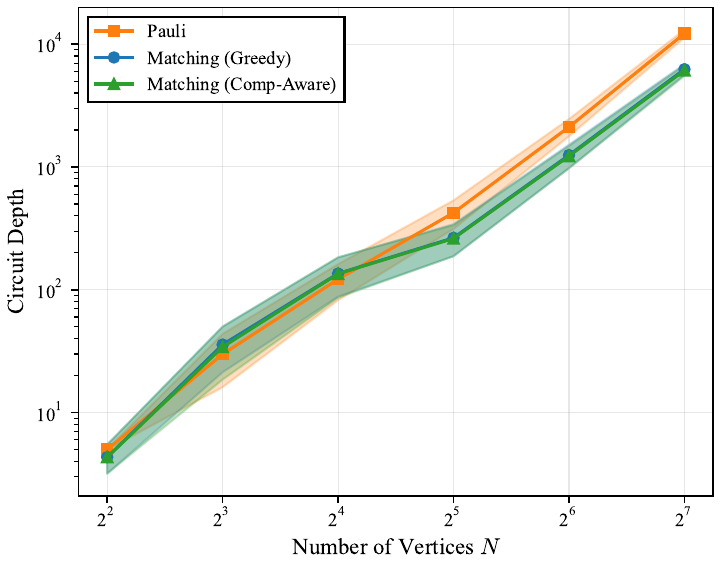}
  \caption{$p = 0.01$}
\end{subfigure}
\caption{CX count and circuit depth for greedy matching decomposition (blue), compression-aware matching decomposition (green), and Pauli decomposition (orange) on Erd\H{o}s--R\'{e}nyi graphs
  $G(N, p)$ with $p=.01$. Shaded regions
  show $\pm 1\sigma$ across 100 graphs. Matching decomposition requires 25\%, 33\%, and 31\% fewer CX gates at N = 32, 64, and 128 respectively, and produces 37\%, 41\%, and 49\% shallower circuits.}
\label{fig:er_cx_scaling}
\end{figure}

The Erd\H{o}s-R\'{e}nyi datasets indicate that sparseness of a graph $G$ may be a property that indicates if matching decomposition requires fewer CX gates than Pauli decomposition when simulating $G$ in the circuit model. However, the connected dataset indicates that matching decomposition can require fewer CX gates than Pauli decomposition on denser graphs if the denser graphs satisfy certain properties such as distribution of edges with particular Hamming weights. 
Interestingly, the connected datasets achieve larger CX reductions than Erd\H{o}s-R\'{e}nyi graphs but exhibit higher variance. 
The connected graphs were designed based on the hypothesis that Hamming weight distribution is the primary structural factor determining matching decomposition performance. However, the persistently high variance suggests that Hamming weight distribution alone does not fully explain performance---other structural properties likely also play significant roles. Future work could investigate these additional structural factors to develop graph generation methods that achieve both large CX reductions and low variance.

\subsection{Commutativity}

We now provide a class of simple graphs where 1.) there exists a matching decomposition such that the matchings commute and 2.) the standard Pauli decomposition has terms that do not pairwise commute. Let us define a \textit{non-commuting Pauli decomposition} to be the standard Pauli decomposition that has terms that do not all pairwise commute and let us define \textit{commuting matching decomposition} of a subgraph $G'$ as a decomposition of the edges of $G'$ into matchings such that the adjacency matrix of each matching of $G'$ pairwise commute. We emphasize that the existence of a commuting matching decomposition does not imply that our matching decomposition algorithm will find this matching.

First, let us consider when the adjacency matrices of two arbitrary subgraphs commute. 
\begin{lemma}
\label{lemma:commuting_subgraphs}
    Let $G = (V,E)$ be a simple, undirected graph. Two arbitrary subgraphs $G' = (V, E')$ and $G'' = (V, E'')$ of $G$ commute if and only if 
    for any $(u, v) \in V \times V$, 
    \vspace{-0.75\baselineskip}
    \begin{align*}
        |\{w \in V\ |\ (u, w) \in E' \text{ and } (w, v) \in E''\}| &= \\ |\{w \in V\ |\ (u, w) \in E'' \text{ and } (w, v) \in E'\}|
    \end{align*}
\end{lemma}
\begin{proof}
    Let $A'$ be the adjacency matrix of $G'$ and $A''$ the adjacency matrix of $G''$.
    Consider matrix element $[u, v]$ of their product:
    \begin{equation*}
        (A'A'')_{uv} = \sum_{w \in V} A'_{uw} A''_{wv}
    \end{equation*}
    Each summand $A'_{uw} A''_{wv}$ is equal to 1 when $(u, w) \in E'$ and $(w, v) \in E''$, therefore $(A'A'')_{uv}$ counts the number of length-2 walks where the first edge belongs to $E'$ and the second edge belongs to $E''$. 
    $(A''A')_{uv}$ does the same except the first edge belongs to $E''$ and the second to $E'$.
    Therefore, the subgraphs commute if and only if the number of these walks is the same regardless of order of subgraphs.
\end{proof}

Less formally, Lemma~\ref{lemma:commuting_subgraphs} requires that the number of length-2 walks between arbitrary 2 vertices of type ($E'$, $E''$) (i.e. where the first edge belongs to $E'$ and the second edge belongs to $E''$) is the same as the number of walks of type ($E''$, $E'$) between the same vertices, so the walk types are balanced.

In the particular case when the subgraphs are matchings, Lemma~\ref{lemma:commuting_subgraphs} is reduced to the following lemma.

\begin{lemma}
\label{lemma:commuting-matchings}
    Two matchings $G' = (V, E')$ and $G'' = (V, E'')$ of a graph $G$ commute if and only if each connected component of their union is either an isolated vertex ($K_1$), single edge ($K_2$) or a 4-cycle ($C_4$).
\end{lemma}
\begin{proof}
    The degree of each vertex in a matching is either 0 or 1. 
    Therefore, in a union of 2 matchings the degree of each vertex is $\le 2$. 
    Therefore, structurally, the union of 2 matchings consists of isolated vertices, paths or even-length cycles where edges alternate between the two matchings.
    Isolated vertices and non-overlapping single edges do not change the number of length-2 walks in the graph, so they do not affect commutativity according to Lemma~\ref{lemma:commuting_subgraphs}.
    If there is a path that consists of more than 1 edge, then we can always consider its first and third vertices. 
    There is only 1 path of length-2 between these vertices, therefore the matchings would not commute according to Lemma~\ref{lemma:commuting_subgraphs}.
    In case of a cycle, if the length of the cycle is more than 4, then we can apply the same argument as for a path to conclude that these matchings do not commute.
    The only case when they still commute is if the cycle length is either 2 (single overlapping edge) or 4, in which case the number of alternating length-2 walks of each type between any 2 vertices of the cycle is balanced (either 0 or 1).
\end{proof}

Our next result shows that a commuting matching remains commuting under relabeling of vertices. We emphasize that we are using a standard Pauli decomposition pipeline and that there may exist more advanced Pauli decomposition techniques 
that allow one to exactly implement the below unitaries without using Trotterization. 

\begin{lemma}\label{lemma:commutingMDecompPreserved}
    A commuting matching decomposition of the adjacency matrix $A$ of a graph is still commuting after a relabeling of the vertices.
\end{lemma}
\begin{proof}
    A relabeling of vertices is a bijection $f(x): \{0, 1\}^n\rightarrow \{0,1\}^n$. Since it is a bijection of the standard basis states, $f(x)$ corresponds to a permutation unitary transformation $U_f$ of the adjacency matrix. Let $M=\{M_k\}$ be a commuting matching decomposition. Then $A=\sum_k A_{M_k}$, where $A_{M_k}$ is the adjacency matrix of $M_k$. Let $A'=U_fAU_f^\dagger$ . First note that $A'$ is an adjacency matrix since it is still 0-1 valued and is symmetric since $f(x)$ is a bijection. Furthermore, $[U_fA_{M_k}U_f^\dagger, U_fA_{M_l}U_f^\dagger ]=0$ for all $k\neq l$. Then $M'=\{f(M_k)\}$ is a commuting matching decomposition of $A'$.
\end{proof}

Consider the adjacency matrix $A=\sum_i^n X_i$. This corresponds to the discrete hypercube graph $Q_n$, which has $2^n$ vertices labeled by all unique n-bitstrings and edges connecting vertices that are Hamming distance one apart.
\begin{lemma}\label{lemma:hyper3Relabling}
    Consider the discrete hypercube with eight vertices, which has adjacency matrix $A=\sum_{i=1}^3 X_i$. Let $f_n(x)=3x \mod(2^n)$. The permutation $f_3(x)$ results in an adjacency matrix $A'$ with a commuting matching decomposition and non-commuting Pauli decomposition (contains $IXI$ and $IYY$).
\end{lemma}
\begin{proof}
    For the initial adjacency matrix $A$, there exists a matching decomposition $M = \sum_{j=1}^n M_j$, where each matching $M_j$ corresponds to all edges that differ in exactly bit $j$.
    From Lemma~\ref{lemma:commuting-matchings}, all $M_j$ are pairwise commuting since their edges form 4-cycles, so $M$ is a commuting matching decomposition.
    From Lemma \ref{lemma:commutingMDecompPreserved}, $M'=f_3(M)$ is a commuting matching decomposition for $A'$.
    
    From direct calculation,
    \begin{align}
        U_f=
        \begin{pmatrix}
        1 & 0 & 0 & 0 & 0 & 0 & 0 & 0\\
        0 & 0 & 0 & 1 & 0 & 0 & 0 & 0\\
        0 & 0 & 0 & 0 & 0 & 0 & 1 & 0\\
        0 & 1 & 0 & 0 & 0 & 0 & 0 & 0\\
        0 & 0 & 0 & 0 & 1 & 0 & 0 & 0\\
        0 & 0 & 0 & 0 & 0 & 0 & 0 & 1\\
        0 & 0 & 1 & 0 & 0 & 0 & 0 & 0\\
        0 & 0 & 0 & 0 & 0 & 1 & 0 & 0
        \end{pmatrix}
    \end{align}
    and
    \begin{align}
        A'=
        \begin{pmatrix}
        0 & 0 & 0 & 1 & 1 & 0 & 1 & 0\\
        0 & 0 & 0 & 1 & 0 & 1 & 1 & 0\\
        0 & 0 & 0 & 0 & 1 & 1 & 1 & 0\\
        1 & 1 & 0 & 0 & 0 & 0 & 0 & 1\\
        1 & 0 & 1 & 0 & 0 & 0 & 0 & 1\\
        0 & 1 & 1 & 0 & 0 & 0 & 0 & 1\\
        1 & 1 & 1 & 0 & 0 & 0 & 0 & 0\\
        0 & 0 & 0 & 1 & 1 & 1 & 0 & 0
        \end{pmatrix}.
    \end{align}
    Then $\tr(IXI A')=4$ and $\tr(IYY A')=-4.$ Thus, the Pauli decomposition of $A'$ contains $IXI$ and $IYY$ which anti-commute. 
\end{proof}
We can easily extend this to an infinite family of hypercubes.

\begin{corollary}[Family of hypercubes]\label{corr:FamilyOfHypercubes}
    Let $A_n=\sum_i^nX_i$ be the adjacency matrix of $Q_n$. Index the labels of the vertices as $b_{n-1}b_{n-2}\cdots b_{0}.$ Then the adjacency matrix $A'_n$ associated with the permutation relabeling
    \begin{align}\label{eq:permutation}
        g_i(b_{n-1}b_{n-2}\cdots b_{0})=b_{n-1} \ldots f_3(b_{i+2}b_{i+1}b_{i})\ldots b_{0} 
    \end{align}
    for any $i$ between $0$ and $n-3$ has a commuting matching decomposition and non-commuting Pauli decomposition. 
\end{corollary}

\begin{proof}
    $A'$ has a commuting matching decomposition by Lemma~\ref{lemma:commuting-matchings} and Lemma~\ref{lemma:commutingMDecompPreserved}. 
    Function $g_i$ induces a unitary $U=I^{\otimes (n-i-3)}\otimes U_f\otimes I^{\otimes i}$. Next, $A'=UAU^\dagger=\sum_j^nUX_jU^\dagger.$ This unitary only acts non-trivially on qubits $i, i+1$, and $i+2$ by definition of Eq.~\eqref{eq:permutation}. From Lemma \ref{lemma:hyper3Relabling}, we know that  $X_{i+1}$ and $Y_{i+2}Y_{i+1}$ are terms in the Pauli decomposition of $A'$, which anti-commute. 
\end{proof}
The previous results hold for an infinite collection of graphs with specific structure. The remainder of this section will provide results for graphs with less rigid structure. If a collection of matchings $\{M_j\}_{j = 1}^k$ satisfies 
\begin{equation*}
    A_G = A_{M_1} + \ldots + A_{M_k}
\end{equation*}
where $A_G$ is the adjacency matrix of some graph $G$, and $A_{M_j}$ is the adjacency matrix of matching $j$, $e^{-i A_Gt} = e^{-i \sum_{j=1}^m A_{M_j} t}$. If all $A_{M_j}$ pairwise commute, then  $e^{-i \sum_{j=1}^m A_{M_j} t} = \prod_{j=1}^k e^{-i A_{M_j} t}$. Thus, a CTQW on graph $G$ is exactly implementable using matching decompositions if and only if the adjacency matrices of all matchings in $\{M_j\}_{j=1}^k$ pairwise commute. 
For non-commuting operators we have the approximation
$$(e^{-i \frac{t}{N}A_{M_1}} e^{-i \frac{t}{N}A_{M_2}} \ldots e^{-i \frac{t}{N}A_{M_k}})^N = e^{-i t A_G} + O(t^2/N)$$
with the error bound depending on the commutator norms $\|[A_{M_i}, A_{M_j}]\|$.

A result similar to Lemma~\ref{lemma:commuting-matchings} has also been proven for perfect matchings in \cite{akbari2007commuting}.

\begin{lemma} 
\cite{akbari2007commuting}
\label{lemma:commute}
    Two perfect matchings of a graph G commute if and only if each connected component of their union is either a single edge ($K_2$) or a 4-cycle ($C_4$). In particular, if a graph $G$ of order $n$ has two disjoint commuting perfect matchings $M1$ and $M2$, then each connected component of $M1 \cup M2$ is a copy of $C_4$, and thus 4 divides $n$.
\end{lemma}

The authors of \cite{akbari2007commuting} also remark:
\begin{remark} \cite{akbari2007commuting}
    By the previous lemma, we conclude that for a k-regular graph $G$, $G$ is decomposable into commuting perfect matchings if and only if $G$ has a 1-factorization such that the union of any two distinct perfect matchings in the 1-factorization is a disjoint union of copies of $C_4$.
\end{remark}
Here, a 1-factorization of a $k$-regular graph $G$ is defined as a partition of the edges into $k$ perfect matchings. It is known that discrete hypercubes, $Q_n$, satisfy this property \cite{behague2019semi, laufer1980strongly}. $Q_n$ is a graph on $2^n$ vertices such that each vertex has an $n$ bitstring label, and two vertices $u$ and $v$ are adjacent if and only if the Hamming distance between their labels is exactly 1. Interestingly, $Q_n$ is a type of circulant graph, the walks on which are known to be exactly implementable in the circuit model \cite{loke2017efficient}. However, this approach requires the Quantum Fourier Transform (QFT) to change the computational basis, along with multi-controlled rotation gates. The matching decomposition approach, in general, appears to require fewer controlled gates to implement walks on $Q_n$ than the QFT approach. 

However, it is not clear which graphs satisfy this property in general. One result from \cite{akbari2009commutativity} shows that the result holds for complete bipartite graphs if each partition has $2^n$ vertices for some $n$.
\begin{lemma}\cite{akbari2009commutativity}
     Let $n$ be a natural number. The graph $K_{n,n}$ is decomposable into commuting perfect matchings if and only if $n$ is a power of 2.
\end{lemma}

Thus, a characterization of the graphs that satisfy Lemma~\ref{lemma:commute} will also provide a set of graphs $G$ such that the CTQWs on $G$ can be exactly implemented using matching decomposition.

\section{Discussion}\label{sec:discussion}
In this work, we develop a new matching heuristic that requires only polynomial complexity classical overhead and can 
be used as a basis for Trotterization when approximating CTQWs in the circuit model.  
We use the adjacency matrices of these matchings as Hamiltonians in a Trotterization of $e^{-i A_G t}$ for graphs $G$ with adjacency matrix $A_G$ and compare these methods to Trotterization with Pauli decomposition \cite{Pauli}. While both matching decomposition and Pauli decomposition have similar quantum gate complexity, matching decomposition achieves lower CX gate counts when used to simulate CTQWs on sparse graphs, as demonstrated in our experiments.

We first find that the 2-norm operator difference between both matching decompositions and the CTQW evolution on sets of 8-, 16-, and 32-vertex graphs is comparable to the 2-norm operator difference between Pauli decomposition and the CTQW evolution. We also simulate CTQWs on sets of connected graphs with between 8 and 128 vertices, as well as Erd\H{o}s-R\'enyi graphs with between 8 and 128 vertices with $p=.01$, and compare the gate counts of the matching decomposition to the gate counts for Pauli decomposition when implemented with Qiskit using circuit optimization level 3 on the same datasets. We find that the greedy matching decomposition requires approximately 34\%, 42\%, and 45\% fewer CX gates than Pauli decomposition respectively on the connected 32-, 64-, and 128-vertex graphs and the compression-aware matching requires approximately .5\%, 39\%, 56\%, and 70\% fewer CX gates than Pauli decomposition on the connected 16-, 32-, 64-, and 128-vertex datasets when implementing matching decomposition. 

While these studies give evidence that matching decomposition can efficiently simulate CTQW Hamiltonian evolution, there are still several open research directions. First, more precise bounds for the CX gate scaling of these approaches would validate the utility of this approach for near-term quantum devices. Second, previous works show that the order of graph matching techniques for state preparation can greatly impact CX gate count \cite{gonzales2025efficient}, so developing heuristics for ordering the matchings to be used in Trotterization can further decrease the gate count. Similarly, randomizing the gates based on the weight of terms in the Hamiltonian can lead to shallower circuits \cite{campbell2019random}. 

Third, a better characterization of when matching decomposition requires fewer gates and shorter circuits than Pauli decomposition would be useful for simulating CTQW Hamiltonian evolution in the circuit model. Fourth, CTQWs are well-known for finding marked nodes in a graph. Can the matchings approaches be used to find marked vertices with high probability in shorter time than the full CTQW Hamiltonian evolution can find them? 
Fifth, it is worth noting that the labels of the nodes may impact circuit depth, thus finding permutation unitaries that reduce gate count further may be useful  \cite{gaidai2026decomposition}. Finally, while matchings are a natural method for decomposing graphs and walks on them are easily implementable in the circuit model, it would be interesting to study if there are other useful graph structures that can be implemented exactly in the circuit model.

Our compression-aware matching algorithm recognizes when multi-controlled gates share the same control pattern, allowing some controls to merge or be eliminated. Pauli decomposition does not inherently recognize controlled gates, which we believe is the source of the gap we observe in the experiments.

\section*{Author contributions}
MA developed the CTQW graph compression technique, wrote code, and gathered data. DD debugged errors and wrote code to analytically verify unitary equivalences. 
IG wrote the proof of Lemma~\ref{lemma:commuting_subgraphs} and \ref{lemma:commuting-matchings}.
AG wrote code, found exact examples, and developed the theoretical idea on the weight-reducing qubits in the edge compression. 
ZHS suggested the experiments to test the algorithm and secured funding. 
RH developed the graph matching decomposition and secured funding. All authors wrote, read, edited, and approved the final manuscript.

\section*{Acknowledgments}
M. Atallah, I. Gaidai, and R. Herrman acknowledge DE-SC0024290. D. Dilley, A. Gonzales, and Z. Saleem acknowledge DOE-145-SE-14055-CTQW-FY23. The funder played no role in study design, data collection, analysis and interpretation of data, or the writing of this manuscript. 

\section*{Competing interests}
All authors declare no financial or non-financial competing interests. 

\section*{Code and Data Availability}
\label{sec:codeAndDataAvail}
The code and data for this research can be found at \url{https://github.com/Mostafa-Atallah2020/dyn-CTQW}. 

\vspace{20pt}
\noindent
\framebox{\parbox{\linewidth}{
The submitted manuscript has been created by UChicago Argonne, LLC, Operator of 
Argonne National Laboratory (``Argonne''). Argonne, a U.S.\ Department of 
Energy Office of Science laboratory, is operated under Contract No.\ 
DE-AC02-06CH11357. 
The U.S.\ Government retains for itself, and others acting on its behalf, a 
paid-up nonexclusive, irrevocable worldwide license in said article to 
reproduce, prepare derivative works, distribute copies to the public, and 
perform publicly and display publicly, by or on behalf of the Government.  The 
Department of Energy will provide public access to these results of federally 
sponsored research in accordance with the DOE Public Access Plan. 
http://energy.gov/downloads/doe-public-access-plan.}}

\bibliographystyle{unsrt}
\bibliography{refs}

\appendix
\onecolumngrid
\newpage

\section{Operator norm convergence plots}\label{sec:normplots}
Figures~\ref{fig:validation_plots_appendix} 
shows the operator 2-norm difference between the exact CTQW Hamiltonian and the greedy and compression-aware matching decomposition and Pauli decomposition approximations. Both methods exhibit comparable convergence rates, with errors scaling as $O(t^2/N)$.
\begin{figure*}[h]
   \centering
    \begin{subfigure}[t]{0.5\linewidth}
      \includegraphics[width=\linewidth]{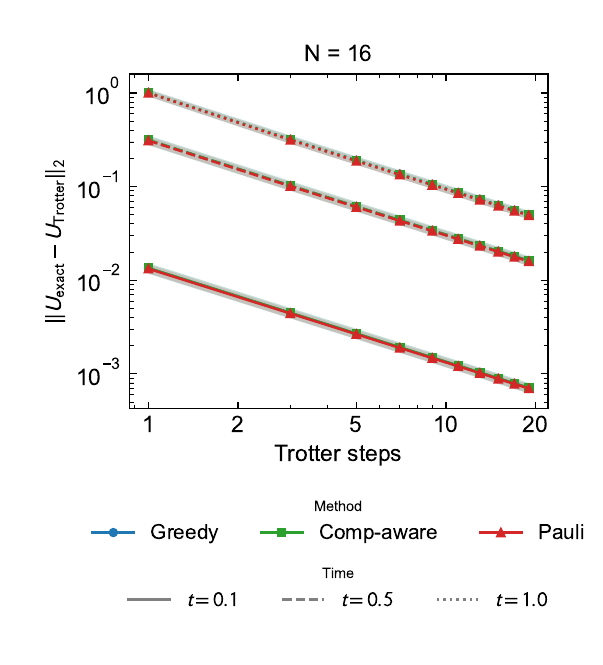}
      \caption{}\label{fig:comparison16vertex}
    \end{subfigure}%
    \begin{subfigure}[t]{0.5\linewidth}
      \includegraphics[width=\linewidth]{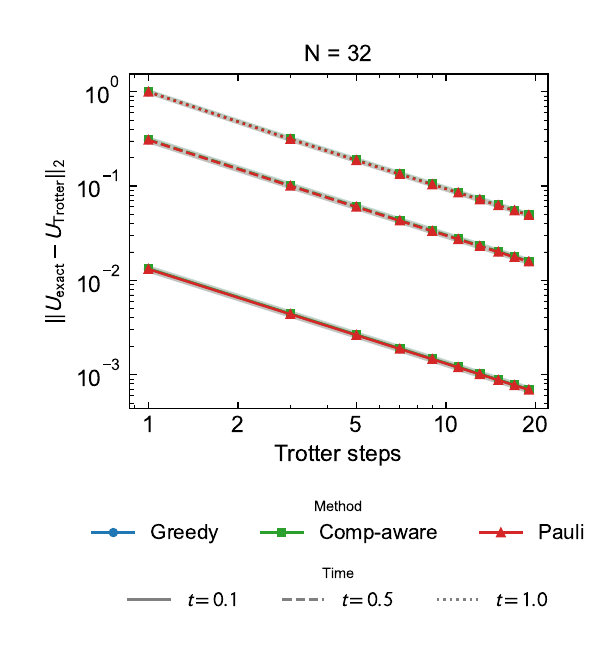}
      \caption{}\label{fig:comparison32vertex}
    \end{subfigure}
    \caption{Operator difference 2-norm $\|e^{-iAt} - U_{\text{Trotter}}\|_2$ between the exact CTQW evolution and Trotterized approximations for sparse connected graphs: (a) 200 graphs on 16 vertices, and (b) 200 graphs on 32 vertices. Lines show mean values over all graphs for greedy matching decomposition (blue, circles), compression-aware matching decomposition (green squares) and Pauli decomposition (red triangles), with shaded regions indicating $\pm 1$ standard deviation. Line styles distinguish evolution times: solid ($t=0.1$), dashed ($t=0.5$), and dash-dot ($t=1.0$). Trotter steps range from $N=10$ to $N=100$.}\label{fig:validation_plots_appendix}
\end{figure*}


\end{document}